\documentclass[journal]{IEEEtran}
\usepackage{amsmath, amssymb}
\usepackage{graphicx, subfigure, pifont, relsize, float, paralist}

\newtheorem{theorem}{Theorem}
\newtheorem{lemma}[theorem]{Lemma}
\newtheorem{proposition}[theorem]{Proposition}

\newtheorem{definition}[theorem]{Definition}

\newtheorem{remark}{Remark}

\newcommand{\myqed}{\hfill $\blacktriangle$%
                   }
\newcommand{\N}{\mathbb N}
\newcommand{\sml}[1]{\mathsmaller{#1}}

\newcommand{\deriv}[2]{ \dot{#1}_{\mathsmaller{#2}} }
\newcommand{\derivsec}[2]{ \ddot{#1}_{\mathsmaller{#2}} }

\begin{document}

\title{Zero-Error Capacity of a Class of Timing Channels}%

\author{
        Mladen~Kova\v cevi\' c,~\IEEEmembership{Student Member,~IEEE,}
        and~Petar~Popovski,~\IEEEmembership{Senior Member,~IEEE}%
\thanks{Date: \today.
        }%
\thanks{M. Kova\v{c}evi\'c is with the Department of Electrical Engineering,
        University of Novi Sad, Serbia (e-mail: kmladen@uns.ac.rs).}%
\thanks{P. Popovski is with the Department of Electronic Systems,
        Aalborg University, Denmark (e-mail: petarp@es.aau.dk).}%
\thanks{M. Kova\v{c}evi\'c was supported by the Ministry of Education,
        Science and Technological Development of the Republic of Serbia 
        (grants TR32040 and III44003).
        P. Popovski was partially supported by the Danish Council for
        Independent Research within the Sapere Aude Research Leader program,
        Grant No. 11-105159, "Dependable wireless bits for M2M communication".
        Part of the work was done while M. Kova\v{c}evi\'c was visiting
        Aalborg University, Denmark, under the support of the EU COST Action IC1104.}%
}%


\maketitle

\begin{abstract}
We analyze the problem of zero-error communication through timing channels
that can be interpreted as discrete-time queues with bounded waiting times.
The channel model includes the following assumptions:
\begin{inparaenum}[1)]
\item
 Time is slotted,
\item
 at most $ N $ ``particles'' are sent in each time slot,
\item
 every particle is delayed in the channel for a number of slots chosen randomly
 from the set $ \{0, 1, \ldots, K\} $, and
\item
 the particles are identical.
\end{inparaenum}
It is shown that the zero-error capacity of this channel is $ \log r $, where
$ r $ is the unique positive real root of the polynomial $ x^{K+1} - x^{K} - N $.
Capacity-achieving codes are explicitly constructed, and a linear-time decoding
algorithm for these codes devised.
In the particular case $ N = 1 $, $ K = 1 $, the capacity is equal to
$ \log \phi $, where $ \phi = (1 + \sqrt{5}) / 2 $ is the golden ratio,
and the constructed codes give another interpretation of the Fibonacci sequence.
\end{abstract}%

\section{Preliminaries}

\IEEEPARstart{T}{he} study of timing channels, i.e., channels that arise when
the information is being encoded in the transmission times of messages, has
resulted in many interesting and relevant models.
Two important and relatively recent examples are the models adopted from queuing
theory \cite{anantharam, bedekar, thomas} and those that arise in molecular
communications \cite{kadloor}.
We analyze here the problem of \emph{zero-error} communication over certain
channels of this type.
The study is partly motivated by settings in which the communication is
done with rather unconventional physical carriers, such as particles, molecules,
items, etc.
These channels can also be viewed as discrete-time queues with bounded waiting
times, and the results can thus be seen as supplementing in a sense the work
carried out in \cite{bedekar,thomas} (see also \cite{prabhakar, nakibly});
however, due to the combinatorial nature of zero-error information theory
\cite{shannon, korner}, the methods used are quite different from those in
\cite{bedekar,thomas}.

\subsection{The channel model}
\label{model}

We assume that multiple transmissions can occur at the same time instant
without interfering with each other.
In this regard, we will use the term \emph{particle} (instead of \emph{symbol}
or \emph{packet}) for the unit of transmission.
We believe that this convention will make the discussion clearer.

Let $ \N $ denote the set of nonnegative integers $ \{ 0, 1, \ldots \} $.
\begin{definition}
The Discrete-Time Particle Channel with parameters $ N, K \in \N $, denoted
DTPC($ N, K $), is the communication channel described by the following assumptions:
\begin{enumerate}
\itemsep0em
  \item Time is slotted, meaning that the particles are sent and received in
        integer time instants;
  \item At most $ N $ particles are sent in each time slot;
  \item Every particle is delayed in the channel for a number of slots chosen
        randomly from the set $ \{0, 1, \ldots, K\} $;
  \item The particles are indistinguishable, and hence the information
        is conveyed via timing only, or equivalently, via the number of
        particles in each slot.   \myqed
\end{enumerate}
\end{definition}

\begin{figure}[h]
 \centering
  \includegraphics[width=\columnwidth]{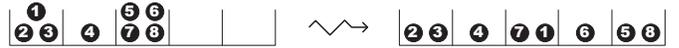}
\caption{Illustration of the DTPC($ 4, 2 $). The particles are numbered only
         for the purpose of illustration, they are assumed identical.}
\label{fig:channel}
\end{figure}%

We elaborate briefly on the definition of the DTPC.
If the duration of the transmission is $ n $ slots, then the assumption 4)
implies that the sequence of particles can be identified with an $ n $-tuple
of integers $ (x_1, \ldots, x_n) \in \{ 0, 1, \ldots, N \}^n $, where
$ x_i $ represents the number of particles in the $ i $'th slot.
For example, Figure \ref{fig:channel} illustrates a situation where the
transmitted sequence is $ ( 3, 1, 4, 0, 0 ) $ and the received sequence is
$ ( 2, 1, 2, 1, 2 ) $.
Hence, the DTPC can be defined purely in terms of sequences of nonnegative
integers, and in the rest of the paper we will rely entirely on this representation.

As for the assumption 3), observe that if the delays of the particles were unbounded
(as is the case, e.g., in queues having service times with geometric distribution
\cite{bedekar}), the zero-error capacity would be zero.
Therefore, in order to obtain interesting models, some restrictions on the delays
have to be imposed.
Similarly, if there is no restriction on the number of particles sent in each
slot, then the zero-error capacity is infinite for any $ K \in \N $, which justifies
the assumption 2).

Note that we have not imposed a restriction on the number of particles at the
output of the DTPC($ N, K $) in a single slot (though it is obviously bounded
by $ (K+1)N $).
It is not hard to argue that this does not affect the zero-error capacity of
the above channel, i.e., it would be the same if this number were also bounded
by $ N $.
This is proven in Appendix \ref{app_Nparticlesreceived}.

Let us also give several more concrete interpretations of the DTPC.
Namely, the ``particles" referred to in the definition of this channel can
be interpreted in various ways depending on the context, e.g., as:

\begin{itemize}
\item
 ``Molecules" in the so-called molecular communications, where the transmission
of information via the number of molecules and their emission times is considered.
The molecules are usually assumed identical, and their arrival times are random
due to their interaction with the fluid medium.
The codes described in the present paper are relevant precisely for the channels
of this type, at least in discrete-time models \cite{kadloor}.
\item
 ``Customers" in queuing systems, an important example of which are queues of
``packets" formed in network routers (see the discussion in Remark \ref{rem:queues}
below).
\item
 ``Packets'' in channels introducing random delays (caused by effects different
than queuing).
Note that the packets referred to in this and the previous paragraph are not
identical in practice and usually carry information via their contents.
In this paper we will be interested in the transmission of information via
\emph{timing only}, similarly as in \cite{bedekar}.
Alternatively, one can imagine a receiver that is not processing the packets
(e.g., a low power node in a wireless sensor network), but only infers their
arrivals through energy detection.
\item
 ``Energy quanta" in a simultaneous transmission of energy and information \cite{petar2}.
\end{itemize}

\begin{remark}[DTPC vs.\ Discrete-Time Queues]
\label{rem:queues}
We have pointed out already that the results of the paper apply also to
queuing systems of certain type.
We introduce them here in a bit more detail.
Denote by DTQR($ N, K $) the Discrete-Time Queue%
\footnote{\,Service procedure, service time distribution, and interarrival
distribution are irrelevant in this context and hence are not specified.}
with $ N $ servers/processors (meaning that $ N $ particles can be processed
simultaneously), with at most $ N $ arrivals per slot, and with Residence
times bounded by $ K $ slots (the residence time of the particle is the total
time that it spends in the queue, either waiting to be processed or being processed).
It is not difficult to argue that the DTPC($ N, K $) and the DTQR($ N, K $)
have identical zero-error codes and zero-error capacities.
The key difference between these channels is that the delays of the particles
in the DTPC are independent, while in the DTQR they are not as they are affected
also by the service procedure (for example, in FIFO queues the particles cannot
be reordered).
The assumption that the particles are identical, however, makes this difference
irrelevant in the zero-error case.
\myqed
\end{remark}

\subsection{Notation and definitions}
\label{notation}

By a ``sequence'' of length $ n $ over a nonempty alphabet $ \mathbb{A} $ we
mean an $ n $-tuple from $ \mathbb{A}^n $.
When there is no risk of confusion, a sequence $ (x_1, \ldots, x_n) $ will also
be written as $ x_1 \cdots x_n $.
If, for a given channel, the sequence $ \bf x $ at its input can produce the
sequence $ \bf y $ at its output with nonzero probability, then we write
$ {\bf x} \rightsquigarrow {\bf y} $.
For any two sequences $ \bf x $ and $ \bf y $, their concatenation is denoted
by $ {\bf x} \circ {\bf y} $, or sometimes simply by $ {\bf x}\,{\bf y} $.
Also, if $ Z $ is a set of sequences, we let
$ {\bf x} \circ Z = \big\{ {\bf x} \circ {\bf z} : {\bf z} \in Z \big\} $
and $ Z \circ {\bf x} = \big\{ {\bf z} \circ {\bf x} : {\bf z} \in Z \big\} $.
We assume that $ {\bf x} \circ \emptyset = \emptyset \circ {\bf x} = \emptyset $,
and $ {\bf x} \circ \varnothing = \varnothing \circ {\bf x} = {\bf x} $,
where $ \emptyset $ denotes an empty set and $ \varnothing $ an empty sequence.
For a sequence $ \bf x $ and a number $ k \in \N $, $ {\bf x}^k $ will
denote the concatenation of $ k $ copies of $ \bf x $, where it is assumed that
$ {\bf x}^0 = \varnothing $.
The weight of a sequence $ {\bf x} = x_1 \cdots x_n $, $ x_i \in \N $, is defined
as $ wt({\bf x}) = \sum_{i=1}^n x_i $.

A code of length $ n $ for the DTPC($ N, K $) is a subset of $ \{ 0, 1, \ldots, N \}^n $.
Codes will be denoted by calligraphic letters $ {\mathcal C}, {\mathcal D} $,
etc., or $ {\mathcal C}(n), {\mathcal D}(n) $, if their length needs to be
emphasized.
The set of codewords of $ {\mathcal C} $ having prefix $ \bf u $ is denoted by
$ {\mathcal C}^{\bf u} $, and the code obtained by removing this prefix by
$ {\mathcal C}^{\underline{\bf u}} = \{ {\bf v} : {\bf u} \circ {\bf v} \in {\mathcal C} \} $.
Clearly, $ {\mathcal C}^{\bf u} = {\bf u} \circ {\mathcal C}^{\underline{\bf u}} $.

\begin{definition}
\label{def_zeroerror}
  $ \mathcal C $ is said to be a zero-error code for the DTPC if for
any $ m \geq 1 $ and any two distinct sequences
$ {\bf x} = {\bf x}_1 \cdots {\bf x}_m $ and
$ {\bf y} = {\bf y}_1 \cdots {\bf y}_m $,
where $ {\bf x}_i, {\bf y}_i \in {\mathcal C} $,
there exists no sequence $ {\bf z} $ such that both
$ {\bf x} \rightsquigarrow {\bf z} $ and $ {\bf y} \rightsquigarrow {\bf z} $.
\myqed
\end{definition}

In words, no two sequences of codewords of $ \mathcal C $ can produce the
same channel output, and hence there is no confusion about which sequence
was sent.
Note that we demand the distinguishability of \emph{sequences of codewords},
rather that just of \emph{codewords}.
This is necessary in the delay channels.
To illustrate this, let $ {\mathcal C} = \{ 000, 100, 001 \} $ be a code
of length three for the DTPC($ 1, 1 $), introducing delays of at most one slot.
Then it is easy to check that no two codewords can produce the same channel
output, but on the other hand $ 001000 \rightsquigarrow 000100 $, and hence
the sequences of codewords $ 001, 000 $ and $ 000, 100 $ are confusable.
$ {\mathcal C} $ is therefore not a zero-error code.

This problem can easily be circumvented by simply padding each codeword with
$ K $ zeros (empty slots, in the original terminology).
Empty slots at the end of each codeword serve to ``catch'' the particles that
are (potentially) sent in the preceding slots and are (potentially) delayed
in the channel.
In this way these particles do not interfere with the following codeword.

\begin{definition}
\label{simple}
  A code $ {\mathcal C}(n) $ for the DTPC($ N, K $) is said to be \emph{zero-padded}
if all of its codewords end with $ \min\{ n, K \} $ zeros.
\myqed
\end{definition}

Clearly, a zero-padded code is zero-error if and only if for every two distinct
\emph{codewords} $ {\bf x} $, $ {\bf y} $, there exists no sequence $ \bf z $ with
$ {\bf x} \rightsquigarrow {\bf z} $ and $ {\bf y} \rightsquigarrow {\bf z} $.%

\begin{definition}
  The rate of a code $ {\mathcal C}(n) $ for the DTPC($ N, K $) is defined
as $ \frac{1}{n} \log |{\mathcal C}(n)| $.
The zero-error capacity of the DTPC is the supremum of the rates of all
zero-error codes for this channel.
The base of $ \log $ is assumed to be $ 2 $ and hence the rates and capacities
are expressed in bits per time slot.
\myqed
\end{definition}

It is easy to show that this supremum is equal to the $ \limsup $ of the
rates of the largest zero-error codes for the DTPC.
Furthermore, when considering the capacity of the DTPC($ N, K $), there is no
loss in generality to restrict oneself to zero-padded codes, because padding with
a constant number of zeros does not affect the code rate in the asymptotic sense.

\section{Optimal zero-error codes for the DTPC}
\label{sec:codes}

In this section we give two constructions of optimal zero-padded zero-error
codes for the DTPC.
The results have similar flavor to those obtained for some other types of
combinatorial channels, e.g., \cite{ahlswede, kobayashi, zhang, ahlswede2}.

\subsection{Recursive construction}
\label{sec:recursive}

The claim that follows establishes a general property of zero-padded zero-error
codes for the DTPC($ N, K $), from which the construction of optimal codes will
follow in a straightforward way.
It states that such codes can, without loss of generality, be assumed to contain
only codewords with prefixes $ N $ and $ i \circ 0^K $, $ i \in \{ 0, 1, \ldots, N-1 \} $.

\begin{proposition}
\label{CNKoptimal}
  Let $ {\mathcal C}(n) $, $ n > K $, be a zero-padded zero-error code for the DTPC($ N, K $).
Then there exists a zero-padded zero-error code $ {\mathcal D}(n) $ of the same size,
and such that:
\begin{equation}
\label{eq:wlogcode}
  {\mathcal D}
       = {\mathcal D}^{\sml{N}}  \, \cup  \bigcup_{i=0}^{N-1} {\mathcal D}^{i \, \circ \, 0^K}  .
\end{equation}
\end{proposition}
\begin{proof}
  Let $ {\mathcal C}(n) $ be a zero-padded zero-error code for the DTPC($ N, K $).
We will construct $ {\mathcal D}(n) $ by removing the codewords of $ {\mathcal C}(n) $
that do not satisfy the desired form, and add the corresponding codewords that do.
For any codeword of $ {\mathcal C}(n) $ of the form
$ {\bf x} = {\bf u} \circ {\bf v} $, where $ {\bf u} = u_{1} \cdots u_{K+1} $ is
of length $ K + 1 $ and weight $ wt({\bf u}) = q $, we let the corresponding
codeword $ \tilde{\bf x} $ of $ {\mathcal D}(n) $ be specified as follows:
If $ q < N $, then $ \tilde{\bf x} =  q \circ 0^K \circ {\bf v} $,
while if $ q \geq N $, then $ \tilde{\bf x} = N \circ \tilde{\bf u} \circ {\bf v} $,
where $ \tilde{\bf u} = \tilde{u}_{2} \cdots \tilde{u}_{K+1} $ is some sequence of
length $ K $ and weight $ q - N $ satisfying $ \tilde{u}_i \leq u_i $,
$ i = 2, \ldots, K+1 $
(in other words, the prefix of $ \tilde{\bf x} $ is in the latter case constructed
from $ {\bf u} $ by removing $ N - u_1 $ of its particles from slots $ 2, \ldots, K+1 $,
and placing them in the first slot, together with the $ u_1 $ particles that are already
there).
Thus, we can write $ {\mathcal D}(n) = \{ \tilde{\bf x} : {\bf x} \in {\mathcal C}(n) \} $.
It is now not difficult to argue that $ |{\mathcal D}(n)| = |{\mathcal C}(n)| $
and that the fact that $ {\mathcal C}(n) $ is a zero-padded zero-error code implies
that $ {\mathcal D}(n) $ is such a code too.
The key observation is that $ {\mathcal C}(n) $ cannot contain two distinct
codewords of the form $ {\bf x}_1 = {\bf u}_1 \circ {\bf v} $ and
$ {\bf x}_2 = {\bf u}_2 \circ {\bf v} $, where the prefixes $ {\bf u}_1 $,
$ {\bf u}_2 $ are of length $ K + 1 $ and have the same weight,
that is $ wt({\bf u}_1) = wt({\bf u}_2) = q $.
This is because $ {\mathcal C}(n) $ is zero-error, and clearly
$ {\bf x}_1 \rightsquigarrow 0^K \circ q \circ {\bf v} $ and
$ {\bf x}_2 \rightsquigarrow 0^K \circ q \circ {\bf v} $.
\end{proof}

\begin{lemma}
\label{zeroerroriffNK}
  Let $ {\mathcal D}(n) $, $ n > K $, be a zero-padded code of the form
\eqref{eq:wlogcode} for the DTPC($ N, K $).
Then $ {\mathcal D}(n) $ is zero-error if and only if $ {\mathcal D}^{\underline{\sml{N}}}(n-1) $
and $ {\mathcal D}^{\underline{i \, \circ \, 0^K}}(n-K-1) $, $ 0 \leq i < N $, are all zero-error.
\end{lemma}
\begin{proof}
  If $ {\mathcal D}^{\underline{\sml{N}}} $ and $ {\mathcal D}^{\underline{i \, \circ \, 0^K}} $
are zero-error, then so are $ {\mathcal D}^{\sml{N}} $ and $ {\mathcal D}^{i \, \circ \, 0^K} $.
Observe also that no two sequences $ \bf x $, $ \bf y $, such that $ \bf x $ has prefix
$ i \circ 0^K $ and $ \bf y $ either $ j \circ 0^K $ or $ N $ (where
$ i, j \in \{ 0, 1, \ldots, N-1 \} $, $ i \neq j $), can produce the same channel output,
implying that $ {\mathcal D} $ is zero-error.
The opposite direction is also easy.
\end{proof}

The above claims imply that an optimal zero-padded zero-error code of length $ n $ for
the DTPC($ N, K $) can be constructed recursively from the codes of length $ n-1 $ and $ n-K-1 $.
To start the recursion, optimal zero-padded zero-error codes of length $ j \in \{ 0, \ldots, K \} $
are needed, which are trivially $ \{ 0^j \} $.

\begin{theorem}
\label{th:CNKoptimal}
  The largest zero-padded zero-error code for the DTPC($ N, K $), denoted $ {\mathcal C}_\sml{N,K} $,
is given by:
\begin{equation}
\label{eq:optimalcode}
 \begin{aligned}
 {\mathcal C}_\sml{N,K}(n)
    = &\big( N \circ {\mathcal C}_\sml{N,K}(n-1) \big) \, \cup  \\
      &\bigcup_{i=0}^{N-1} \big( i \circ 0^K \circ {\mathcal C}_\sml{N,K}(n-K-1) \big) ,
 \end{aligned}
\end{equation}
for $ n > K $, and $ {\mathcal C}_\sml{N,K}(n) = \{ 0^n \} $ for $ 0 \leq n \leq K $.
\hfill \IEEEQED
\end{theorem}

In the following subsection we will describe a different, perhaps more intuitive
construction of the codes $ {\mathcal C}_\sml{N,K} $.

Theorem \ref{th:CNKoptimal} implies that the cardinalities of the codes $ {\mathcal C}_\sml{N,K} $
satisfy the recurrence relation:
\begin{equation}
 \left| {\mathcal C}_\sml{N,K}(n) \right|
     = \left| {\mathcal C}_\sml{N,K}(n-1) \right| + N \left| {\mathcal C}_\sml{N,K}(n-K-1) \right| ,
\end{equation}
with initial conditions $ \left| {\mathcal C}_\sml{N,K}(n) \right| = 1 $, $ 0 \leq n \leq K $,
which further implies that:
\begin{equation}
\label{eq:CNKsize}
 \left| {\mathcal C}_\sml{N,K}(n) \right|
     = \sum_{k=0}^{K+1} a_k r_k^n ,
\end{equation}
where $ r_k $ are the (complex) roots of the polynomial $ x^{K+1} - x^K - N $, and
$ a_k $ are (complex) constants.

\begin{remark}
  In the particular case $ N = 1 $, $ K = 1 $, the analysis of the channel amounts to
analyzing binary sequences whose $ 1 $'s are being shifted in the channel by at most
one position to the right (hence, the DTPC($ 1, 1 $) can also be seen as a type of a
``bit-shift'' channel \cite{shamai, krachkovsky}).
In this case, the codes $ {\mathcal C}_\sml{1,1} $ satisfy the relation%
\footnote{\,As one of the referees pointed out, this resembles a well known
characterization of $ F_n $ as the number of binary sequences of length $ n - 1 $
with no consecutive ones.
Such a set of sequences, $ {\mathcal S}(n) $, obeys the recursion
$ {\mathcal S}(n) = \big( 0 \circ {\mathcal S}(n-1) \big) \cup \big( 10 \circ {\mathcal S}(n-2) \big) $,
with $ {\mathcal S}(0) = \{ \varnothing \} $, $ {\mathcal S}(1) = \{ 0, 1 \} $.}:
\begin{equation}
  {\mathcal C}_\sml{1,1}(n) = \big( 1 \circ {\mathcal C}_\sml{1,1}(n-1) \big) \cup \big( 00 \circ {\mathcal C}_\sml{1,1}(n-2) \big) ,
\end{equation}
with $ {\mathcal C}_\sml{1,1}(0) = \{ \varnothing \} $, $ {\mathcal C}_\sml{1,1}(1) = \{ 0 \} $,
which implies that
$ |{\mathcal C}_\sml{1,1}(n)| = |{\mathcal C}_\sml{1,1}(n-1)| + |{\mathcal C}_\sml{1,1}(n-2)| $,
with $ |{\mathcal C}_\sml{1,1}(0)| = |{\mathcal C}_\sml{1,1}(1)| = 1 $.
In other words, $ ( |{\mathcal C}_\sml{1,1}(n)| ) $ is the Fibonacci sequence\footnote{\,The
name Fibonacci code would thus be appropriate here, but it has already been used in
some other contexts \cite{kautz, zhang}.}  $ ( F_n ) $.
\myqed
\end{remark}

\subsection{Direct construction}

Let $ {\mathcal D}_\sml{N,K}(n) $ be the code defined by the following procedure.
First enumerate in the inverse lexicographic order all sequences of length $ n $
over $ \{ 0, 1, \ldots, N \} $ ending with $ \min \{ n, K \} $ zeros (so that,
for $ n > K $, the first sequence on the list is $ N^{n-K} \circ 0^K $, the
second one is $ N^{n-K-1} \circ (N-1) \circ 0^K $, etc.; see Table \ref{tab:code21}).
Then repeat the following step until there are no more sequences to process:
Select the first sequence on the list that has not been processed, call it $ \bf x $,
to be a codeword, and then exclude all sequences $ \bf y $ such that
$ {\bf x} \rightsquigarrow {\bf y} $.
Table \ref{tab:code21} illustrates the construction for $ N = 2 $, $ K = 1 $
(only the codewords are listed to save space).

{
\renewcommand{\arraystretch}{0.95}
\begin{table}[h]
\centering
\caption{Zero-error codes of length up to $ 5 $ for the DTPC($ 2, 1 $).
         Cardinality of the codes is shown in the rightmost column.}
 \begin{tabular}{ | c c c c c r }
  \multicolumn{1}{|c}{2} & \multicolumn{1}{|c}{2} & \multicolumn{1}{|c}{2} & \multicolumn{1}{|c}{2} & \multicolumn{1}{|c}{0}  &  1  \\ \cline{5-6}
  \multicolumn{1}{|c}{2} & \multicolumn{1}{|c}{2} & \multicolumn{1}{|c}{2} & \multicolumn{1}{|c}{1} & 0                       &     \\
  \multicolumn{1}{|c}{2} & \multicolumn{1}{|c}{2} & \multicolumn{1}{|c}{2} & \multicolumn{1}{|c}{0} & 0                       &  3  \\ \cline{4-6}
  \multicolumn{1}{|c}{2} & \multicolumn{1}{|c}{2} & \multicolumn{1}{|c}{1} &                     0  & 0                       &     \\
  \multicolumn{1}{|c}{2} & \multicolumn{1}{|c}{2} & \multicolumn{1}{|c}{0} &                     0  & 0                       &  5  \\ \cline{3-6}
  \multicolumn{1}{|c}{2} & \multicolumn{1}{|c}{1} & 0                      &                     2  & 0                       &     \\
  \multicolumn{1}{|c}{2} & \multicolumn{1}{|c}{1} & 0                      &                     1  & 0                       &     \\
  \multicolumn{1}{|c}{2} & \multicolumn{1}{|c}{1} & 0                      &                     0  & 0                       &     \\
  \multicolumn{1}{|c}{2} & \multicolumn{1}{|c}{0} & 0                      &                     2  & 0                       &     \\
  \multicolumn{1}{|c}{2} & \multicolumn{1}{|c}{0} & 0                      &                     1  & 0                       &     \\
  \multicolumn{1}{|c}{2} & \multicolumn{1}{|c}{0} & 0                      &                     0  & 0                       &  11 \\ \cline{2-6}
  \multicolumn{1}{|c}{1} & 0                      & 2                      &                     2  & 0                       &     \\
  \multicolumn{1}{|c}{1} & 0                      & 2                      &                     1  & 0                       &     \\
  \multicolumn{1}{|c}{1} & 0                      & 2                      &                     0  & 0                       &     \\
  \multicolumn{1}{|c}{1} & 0                      & 1                      &                     0  & 0                       &     \\
  \multicolumn{1}{|c}{1} & 0                      & 0                      &                     0  & 0                       &     \\
  \multicolumn{1}{|c}{0} & 0                      & 2                      &                     2  & 0                       &     \\
  \multicolumn{1}{|c}{0} & 0                      & 2                      &                     1  & 0                       &     \\
  \multicolumn{1}{|c}{0} & 0                      & 2                      &                     0  & 0                       &     \\
  \multicolumn{1}{|c}{0} & 0                      & 1                      &                     0  & 0                       &     \\
  \multicolumn{1}{|c}{0} & 0                      & 0                      &                     0  & 0                       &  21 \\ \cline{1-6}
 \end{tabular}
\label{tab:code21}
\end{table}
}



\begin{proposition}
\label{th:direct}
  $ {\mathcal D}_\sml{N, K}(n) = {\mathcal C}_\sml{N, K}(n) $ for every $ n \in \N $.
\end{proposition}
\begin{proof}
Since $ {\mathcal D}_\sml{N, K}(n) = {\mathcal C}_\sml{N, K}(n) = \{ 0^n \} $ for
$ 0 \leq n \leq K $, it is enough to show that $ {\mathcal D}_\sml{N, K}(n) $
satisfy the relation \eqref{eq:optimalcode}.
Observe that
$ {\mathcal D}_{\sml{N,K}}^\sml{N}(n) = N \circ {\mathcal D}_{\sml{N,K}}(n-1) $
because adding a fixed prefix to a set of sequences does not affect the process of
construction and, moreover, the prefix $ N $ puts the sequences on the top of the list.
It is left to prove that
$ {\mathcal D}_{\sml{N,K}}^{i}(n) = i \circ 0^K \circ {\mathcal D}_{\sml{N,K}}(n-K-1) $,
for $ 0 \leq i < N $.
First consider the case $ i = N - 1 $.
Let $ {\bf x} $ be a sequence with prefix $ (N-1) \circ {\bf u} $, where
$ {\bf u} $ is of length $ K $ and strictly positive weight, and construct
$ \tilde{\bf x} $ as in the proof of Proposition \ref{CNKoptimal}.
Now, if $ \tilde{\bf x} $ is a codeword,
then $ {\bf x} $ is not because $\tilde{\bf x} \rightsquigarrow {\bf x} $ and so
$ \bf x $ would have been excluded in the process of construction.
On the other hand, if $ \tilde{\bf x} $ is not a codeword, then it has itself
been excluded by some sequence $ \bf y $ that precedes it in the inverse
lexicographic order ($ {\bf y} \rightsquigarrow \tilde{\bf x} $).
But then we also have $ {\bf y} \rightsquigarrow {\bf x} $, and therefore
$ {\bf x} $ is not a codeword either.
%
We have shown that $ {\mathcal D}_{\sml{N,K}}^\sml{N-1}(n) $ does not contain
codewords having prefix $ (N-1) \circ {\bf u} $, where $ wt({\bf u}) > 0 $,
and hence it can only contain codewords starting with $ (N-1) \circ 0^K $.
Since none of the sequences with this prefix could have been excluded in the
process of construction by a codeword from $ {\mathcal D}_{\sml{N,K}}^\sml{N}(n) $,
it follows that they have been processed independently of the rest of the list
and so
$ {\mathcal D}_{\sml{N,K}}^\sml{N-1}(n) = (N-1) \circ 0^K \circ {\mathcal D}_\sml{N,K}(n-K-1) $.
One can now prove by induction that
$ {\mathcal D}_{\sml{N,K}}^{i}(n) = i \circ 0^K \circ {\mathcal D}_\sml{N,K}(n-K-1) $
for $ i = N-1, N-2, \ldots, 1, 0 $;
the argument is very similar to the above and is omitted.
\end{proof}

\subsection{Decoding algorithm}

The structure of the codes $ {\mathcal C}_\sml{N,K} $, captured by the relation
\eqref{eq:optimalcode}, suggests a very simple algorithm for recovering the
transmitted sequence $ {\bf x} = x_1 \cdots x_{n} \in {\mathcal C}_\sml{N,K}(n) $
from the received sequence $ {\bf y} = y_1 \cdots y_{n} $.
The procedure is as follows:

Set $ {\bf y}^\sml{(1)} = {\bf y} $, and
observe the prefix of $ {\bf y}^\sml{(1)} $ of length $ K + 1 $, namely
$ y_1 \cdots y_{K+1} $, and its weight $ q $.
\begin{itemize}
\item[1.]
If $ q < N $, conclude that $ x_1 \cdots x_{K+1} = q \circ 0^K $
(see \eqref{eq:optimalcode}), and set $ {\bf y}^\sml{(2)} = y_{K+2} \cdots y_n $.
Note that $ {\bf y}^\sml{(2)} $ is the (possible) output of the DTPC($ N, K $) when
the input is the codeword $ x_{K+2} \cdots x_n $ from $ {\mathcal C}_\sml{N,K}(n-K-1) $.
\item[2.]
If $ q \geq N $, conclude that $ x_1 = N $.
If also $ y_1 < N $, this means that some of the particles from the first slot
have been delayed in the channel.
In that case remove $ N - y_1 $ of these particles from slots $ 2, \ldots, K+1 $
(first taking particles from slot $ 2 $, then slot $ 3 $, etc., until $ N - y_1 $
of them are collected) and put them in the first slot.
Then set $ {\bf y}^\sml{(2)} = y'_2 \cdots y'_{K+1} \circ y_{K+2} \cdots y_n $,
where $ y'_2 \cdots y'_{K+1} $ is obtained from $ y_2 \cdots y_{K+1} $ by
removing the particles in the above-described way, i.e., for some
$ k \in \{ 2, \ldots, K + 1 \} $ we have $ y'_i = 0 $ for
$ i \in \{ 2, \ldots, k - 1 \} $,
$ y'_k = \sum_{i=1}^{k} y_i - N \geq 0 $, and
$ y'_i = y_i $ for $ i \in \{ k + 1, \ldots, K + 1 \} $.
Note that $ {\bf y}^\sml{(2)} $ is the (possible) output of the DTPC($ N, K $)
when the input is the codeword $ x_2 \cdots x_n \in {\mathcal C}_\sml{N,K}(n-1) $.
\end{itemize}
The procedure is repeated with $ {\bf y}^\sml{(2)} $ by considering its prefix
of length $ K + 1 $, and so on.

Since at least one symbol of $ {\bf x} $ is determined in every iteration, the
algorithm will terminate in at most $ n $ iterations (in fact, at most $ n - K $
due to the trailing zeros).

\section{Zero-error capacity of the DTPC}
\label{sec:capacity}

The results of Section \ref{sec:recursive} imply that the capacity of the
DTPC($ N, K $) can be simply found as
$ \lim_{ n \to \infty } \frac{1}{n} \left| {\mathcal C}_\sml{N, K}(n) \right| $,
and by using the fact that the asymptotic behavior of
$ \left| {\mathcal C}_\sml{N, K}(n) \right| $ is determined by the largest (in
modulus) root of the polynomial $ x^{K+1} - x^K - N $ (see \eqref{eq:CNKsize}).

\begin{lemma}
\label{root}
  The largest (in modulus) root $ r $ of the polynomial $ x^{K+1} - x^K - N $
is real and greater than $ 1 $.
Moreover, if $ K \to \infty $, then $ r \to 1 $.
\end{lemma}
\begin{proof}
  The following theorem is proven in \cite[Ch.\ 3, Thm 2]{wilfbook} (see also \cite{wilf}):
If $ p(x) = c_m x^m + c_{m-1} x^{m-1} + \cdots + c_1 x + c_0 $ is an arbitrary
polynomial with complex coefficients, and $ c_0 \cdot c_m \neq 0 $, then all
roots of $ p(x) $ lie in the (complex) circle $ |x| \leq r $, where $ r $ is
the \emph{unique positive real} root of
$ \tilde{p}(x) = |c_m| x^m - |c_{m-1}| x^{m-1} - \cdots - |c_1| x - |c_0| $.
Since our polynomial is precisely of the form $ \tilde{p}(x) $, we conclude
that it has a unique positive real root $ r $, and that all other roots are
smaller in modulus than $ r $.
This root can be found as the point of intersection of the curves $ x^{K} $
and $ N (x-1)^{-1} $ (viewed as real functions).
By analyzing these curves it follows easily that $ r > 1 $ and that $ r \to 1 $
when $ K \to \infty $.
\end{proof}

\begin{theorem}
\label{th:capacityNK}
  The zero-error capacity of the DTPC($ N, K $) is equal to $ \log r $, where $ r $
is the unique positive real root of the polynomial $ x^{K+1} - x^{K} - N $.
\hfill \IEEEQED
\end{theorem}

The zero-error capacity of the DTPC($ 1, 1 $) is therefore $ \log \phi $, where
$ \phi = (1 + \sqrt{5}) / 2 $.
More generally, the zero-error capacity of the DTPC($ N, 1 $) equals
$ \log \left( \frac{1}{2}(1 + \sqrt{1 + 4N}) \right) $.
Explicit expressions can also be obtained in the following two cases, which are
intuitively clear: The zero-error capacity of the DTPC($ N, 0 $) is $ \log ( N + 1 ) $,
while that of the DTPC($ N, \infty $) (which allows arbitrarily large delays) is
zero.

\begin{figure}[h]
 \centering
  \subfigure[Dependence on $ K $, for $ N = 1, 3, 7, 15, 31, 63 $.]
  {
    \includegraphics[width=\columnwidth,trim=1.5cm 0cm 1cm 0cm]{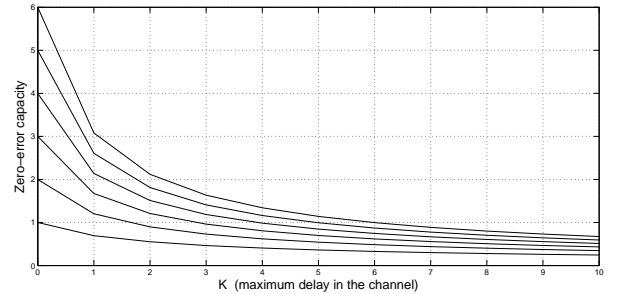}
    \label{fig_capacityN}
  }
  \subfigure[Dependence on $ N $, for $ K = 0, 1, \ldots, 10 $.]
  {
    \includegraphics[width=\columnwidth,trim=1.5cm 0cm 1cm 0cm]{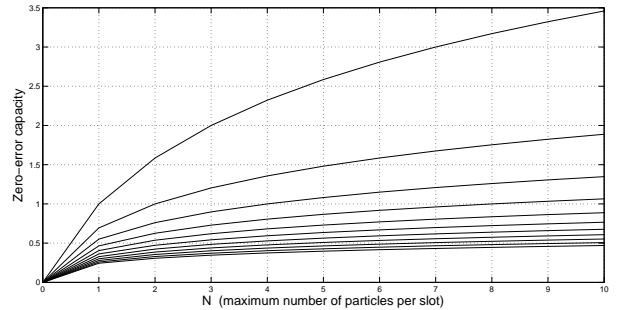}
    \label{fig_capacityK}
  }
\caption{The behavior of the zero-error capacity of the DTPC($ N, K $).}%
\label{fig_capacity}
\end{figure}%

The following proposition states some basic properties of the capacity, regarded
as a function of the channel parameters $ N $ and $ K $.
This function is also illustrated in Figure \ref{fig_capacity}.

\begin{proposition}
  Both $ r $ and $ \log r $ are monotonically increasing concave functions of
$ N $, for fixed $ K $, and monotonically decreasing convex functions of $ K $,
for fixed $ N $.
\end{proposition}
\begin{proof}
  The function $ r $ is defined implicitly by $ r^{K+1} - r^K - N = 0 $, $ r > 1 $,
and the function $ c = \log r $ by $ 2^{c(K+1)} - 2^{cK} - N = 0 $, $ c > 0 $.
Note that $ r $ and $ c $ are well-defined for all $ N, K \in \mathbb{R}_\sml{+} $,
not necessarily integers.
One can therefore differentiate them with respect to $ N $ and $ K $ and verify
that $ \deriv{r}{N} > 0 $, $ \derivsec{r}{N} < 0 $, $ \deriv{c}{K} < 0 $,
$ \derivsec{c}{K} > 0 $.
The remaining claims follow from the properties of the logarithm and the exponential
function.
\end{proof}

%
%

\section*{Acknowledgment}

The authors would like to thank the Associate Editor J\'anos K\"orner and the
anonymous reviewers for providing constructive comments and significantly
improving the quality of the manuscript.


\appendices

\section{Restricting the channel output}
\label{app_Nparticlesreceived}

In this section we demonstrate that bounding the number of particles that
can be received in a slot by $ N $ (or by $ N' \geq N $) does not change
the zero-error capacity of the DTPC.
For the purpose of this argument we will refer to the channel with this
additional restriction as the DTPC($ N, K; N $).
To clarify what is meant by the DTPC($ N, K; N $), we emphasize that there is
no ``limiter" in the channel that drops some of the particles if their number
in a slot exceeds $ N $.
As in the DTPC($ N, K $), all particles must arrive at the destination, only
now their delays, in addition to being $ \leq K $, have to be such that the
number of particles at the channel output in every slot is $ \leq N $.
One can perhaps imagine a ``membrane" at the channel output allowing at most
$ N $ particles per slot to pass through.

\begin{proposition}
\label{prop_Nparticlesreceived}
  Any zero-error code for the DTPC($ N, K $) is a zero-error code for the
DTPC($ N, K; N $), and vice versa.
\end{proposition}
\begin{proof}
  Let $ {\bf x} $ and $ {\bf y} $ be two sequences such that they can both
produce $ {\bf z} = z_1 \cdots z_{l} $ at the output of the DTPC($ N, K $).
Then there exists $ {\bf w} = w_1 \cdots w_{l} $ such that
$ {\bf x} \rightsquigarrow {\bf w} $ and $ {\bf y} \rightsquigarrow {\bf w} $
in the DTPC($ N, K; N $), i.e., such that $ w_i \leq N $.
To see this, observe that if $ z_i > N $ for some $ i \in \{ 1, \ldots, l \} $,
then some of these $ z_i $ particles have not been delayed for a maximal number
of slots ($ K $) and could be further delayed.
We can therefore find the desired $ \bf w $ by going through slots $ 1, \ldots, l $,
respectively, and whenever we find that $ z_i' > N $, we move $ z_i' - N $ of these
particles to slot $ i + 1 $, where $ z_i' $ is the sum of $ z_i $ and the number of
particles that were moved from slots $ 1, \ldots, i - 1$ to slot $ i $ during this
procedure.
We conclude that if a code is \emph{not} a zero-error code for the DTPC($ N, K $),
then it is \emph{not} a zero-error code for the DTPC($ N, K; N $) either.
The opposite direction is obvious.
\end{proof}

We note, however, that bounding the number of received particles in a slot by
$ N' < N $ reduces the zero-error capacity because it excludes some sequences
as valid inputs.


%
%

\end{document}